\pdfoutput=1
\documentclass[sigconf,natbib=false,pbalance,authorversion]{acmart}

\setcopyright{cc}
\setcctype{by}
\copyrightyear{2025}
\acmYear{2025}
\acmDOI{10.1145/3734436.3734444}
\acmConference[SACMAT '25]{Proceedings of the 30th ACM Symposium on Access Control Models and Technologies}{July 8--10,
  2025}{Stony Brook, NY}
\acmBooktitle{Proceedings of the 30th ACM Symposium on Access Control Models and Technologies (SACMAT '25), July 8--10, 2025, Stony Brook, NY, USA}
\acmISBN{979-8-4007-1503-7/2025/07}

\RequirePackage[
datamodel=acmdatamodel,
style=acmnumeric, % use style=acmauthoryear for publications that require it
]{biblatex}
\usepackage{graphicx}
\usepackage{tikz}
\usepackage{amsthm}
\usepackage{algorithm}
\usepackage{algpseudocodex}
\usepackage{balance}

\usepackage{booktabs}

\newcommand{\lep}{\le_p}
\newcommand{\DACC}{\mathrm{DACC}}
\newcommand{\COL}{\mathrm{3COL}}
\newcommand{\SP}{\mathrm{SP}}
\newcommand{\coSP}{\mathrm{coSP}}

\newcommand{\NP}{\mathsf{NP}}
\newcommand{\coNP}{\mathsf{coNP}}

\begin{document}
\title{Safety Analysis in the NGAC Model}

\author{Brian Tan}
\email{Brian.Tan@colostate.edu}
\orcid{0009-0005-2455-8791}
\affiliation{%
 \institution{Colorado State University}
 \city{Fort Collins}
 \state{CO}
 \country{USA}}
 
\author{Ewan S. D. Davies}
\email{Ewan.Davies@colostate.edu}
\orcid{0000-0002-2699-0976}
\affiliation{%
 \institution{Colorado State University}
 \city{Fort Collins}
 \state{CO}
 \country{USA}}
 
\author{Indrakshi Ray}
\email{Indrakshi.Ray@colostate.edu}
\orcid{0000-0002-0714-7676}
\affiliation{%
 \institution{Colorado State University}
 \city{Fort Collins}
 \state{CO}
 \country{USA}}
 
\author{Mahmoud A. Abdelgawad}
\email{M.Abdelgawad@colostate.edu}
\orcid{0000-0002-9407-6342}
\affiliation{%
 \institution{Colorado State University}
 \city{Fort Collins}
 \state{CO}
 \country{USA}}

\begin{abstract}
    We study the safety problem for the next-generation access control (NGAC) model. We show that under mild assumptions it is $\coNP$-complete, and under further realistic assumptions we give an algorithm for the safety problem that significantly outperforms naive brute force search. 
    We also show that real-world examples of mutually exclusive attributes lead to nearly worst-case behavior of our algorithm.
\end{abstract}




\maketitle

\section{Introduction}

Access control policies (ACPs) are an essential technology in cybersecurity. Broadly, given some users and resources, an ACP guarantees the legitimate access of resources by authorized users as well as preventing unauthorized access. 

There are several widely-accepted ACP frameworks including Role-Based Access Control (RBAC) and Attribute-Based Access Control (ABAC). We study a prominent example of ABAC, the Next Generation Access Control (NGAC) model published by the American National Standards Institute~\cite{INCITS}. Roughly, the \emph{state} of an NGAC model is equivalent to a directed graph (digraph) whose vertices comprise sets of users, user attributes, resource attributes, resources, and policy classes. The edges of the digraph correspond to relations on the vertices that fall into one of four categories: obligations, prohibitions, associations, and assignments. Given a state, we say that a user $u$ can \emph{access} a resource $r$ if there is a path from $u$ to $r$ in the associated digraph. 

An important feature of the NGAC model in applications is the provision for \emph{commands} (also known as \emph{obligations}) that, subject to some conditions, change the state of the model. For example, a command could introduce or delete vertices and edges in the digraph after checking a condition such as the existence of an edge in the current state. The dynamic nature of an NGAC model means that formal analysis is nontrivial, and in this work we are interested in a notion of safety for such models. 
For us, an NGAC model is \emph{safe} if the execution of any available obligations cannot modify the state in such a way that a new access path from some user $u$ to a resource $r$ is created. The precise definition is somewhat technical, we defer a formal statement to Section~\ref{sec:prelim}. The \emph{safety problem} ($\SP$) is the problem of determining, given a description of an NGAC model, whether the model is safe in this sense. 
An important subtlety is that safety is represented as the \emph{absence} of some concrete structure in the model, namely a sequence of commands whose execution results in a new access path. 
Thus, it is convenient to consider the complementary problem that could naturally be called the \emph{unsafety problem} or more technically the \emph{co-safety problem} ($\coSP$).

Our main contributions are to give a rigorous definition for the safety problem (and co-safety problem) in the NGAC model, and to show that co-safety is an $\NP$-complete decision problem (in the mono-operational case, see below).
We achieve this in two steps, first abstracting the dynamic nature of co-safety in the NGAC model into a problem of path-finding subject to constraints in a static digraph that we call \emph{directed acyclic constrained connectivity} ($\DACC$). We then reduce 3-coloring ($\COL$) to this decision problem, showing that co-safety is NP-complete via the relation
\[ \COL \lep \DACC \lep \coSP \] 
where $\mathrm A \lep \mathrm B$ represents the standard notion that there is a polynomial-time reduction from problem $\mathrm A$ to problem $\mathrm B$. 
That $\COL$ can be reduced to the co-safety problem in the NGAC model shows that $\coSP$ is $\NP$-hard, and proving that $\coSP\in\NP$ is trivial.
While the fact that a problem equivalent to $\DACC$ is $\NP$-complete has been discussed in the literature before (cf.\ Problem GT54 in Garey and Johnson's comprehensive book~\cite{GJ79} and the paper~\cite{GMO76}), since $\DACC$ is not a well-known $\NP$-complete problem we give the short proof for completeness.

Our secondary contributions are to exploit the graph-theoretic nature of the safety problem for the NGAC model to give an algorithm for the safety problem that significantly outperforms naive brute-force search. Though our algorithm is exponential-time in general, it is adaptive to a graph-theoretic property known as the number of \emph{maximal independent sets} (MIS) in a natural ``supergraph'' constructed from the NGAC model. 
That is, if the specific NGAC model being studied yields supergraph that has a polynomially-bounded number of MIS then our algorithm decides the safety problem in polynomial time. We also investigate the real-world implications of this result by considering natural structures that occur in real-world NGAC models. 
Here, our findings are negative and we give a natural class of NGAC models which encode the structure known by the classic combinatorial result typically attributed to Moon and Moser~\cite{MM60,MM65} to yield the worst possible case for the number of MIS. These observations link the complexity of safety in the NGAC model to a large body of work in extremal combinatorics and the algorithmic study of enumerating MIS.
Writing $\mu(C)$ for the set of MIS in a simple graph $C$, we know~\cite{MM60,MM65} that $|\mu(C)|\le 1.45^{|V(C)|}$, and this forms the basis of the worst-case time complexity of our algorithm.
Naive search for (maximal) independent sets in $C$ takes much longer: time $\Omega(2^{|{V(C)}|})$.

\subsection{Related work}

The kind of safety problem we study here has a long and rich history in cybersecurity. The foundational work of Harrison, Ruzzo and Ullman~\cite{HRU76} defines an ACP framework whose state consists of sets of \emph{rights} in the entries of a matrix indexed by subject-object pairs. We refer to this widely studied model, which appears in many introductory texts~\cite{BB03,Den82,PPC24}, as the HRU model. 
As with the more modern NGAC model, the HRU model is dynamic and features state-changing commands. 

In part due to the flexibility of the HRU model, the safety problem for the model is rather subtle. 
Harrison, Ruzzo and Ullman themselves demonstrated~\cite{HRU76} that different natural formulations of the safety problem for the HRU model can be $\NP$-complete or even undecidable. 
One of the restrictions that permits the $\NP$-completeness proof is that commands are mono-operational (which we explain later), and we make an analogous assumption in our main result too.

Aside from the fact that natural variations on safety can have such wildly different computational complexities, the challenges of safety analysis are further evidenced by the comprehensive review the HRU model due to Tripunitara and Li~\cite{TL13}. Tripunitara and Li refined the definitions of safety problems for the HRU model from the original work to resolve several ambiguities. 
They also corrected several errors in widely-circulated proofs of complexity results on safety problems for the HRU model. 
Their rigorous treatment of the computational complexity of safety for a concrete ACP framework serves as direct inspiration for our work.
In the wake of the original study of Harrison, Ruzzo and Ullman, there are works defining a number of alternative ACP frameworks and a broad literature on safety~\cite{AS91,Bud83,JLS76,LT06,LS77,MPSV00,San92,SS04,Sos00,SMO04}.

\subsection{Organization}

In Section~\ref{sec:prelim} we state precisely what we mean by the NGAC model and define the safety problem that we study. 
In Section~\ref{sec:complexity} we give a precise definition of directed acyclic constrained connectivity ($\DACC$) in order to strip away features of the NGAC model which are not especially important from the perspective of the complexity of the safety problem. 
We then give the reduction from 3-coloring to $\DACC$ and the reduction from $\DACC$ to the co-safety problem. 
In Section~\ref{sec:alg} we give our algorithm for the safety problem that exploits maximal independent sets, as well as proving correctness and running time bounds.
Finally, in Section~\ref{sec:real} we discuss aspects of the performance of our algorithm in a real-world context.

\section{Preliminaries}\label{sec:prelim}
%\subsection{Model Configuration}
An NGAC model $M$ consists of an $11$-tuple 
\[M = ( U, UA, R, RA, R_{\psi}, A_U, A_R, ASC, P, V, COM)\]
where
\begin{itemize}
    \item $U$ is a set of users,
    \item $UA$ is a set of user attributes,
    \item $R$ is a set of resources,
    \item $RA$ is a set of resource attributes,
    \item $R_\psi$ is a set of access rights,
    \item $A_U$ is a set of assignment edges, % (in the ``user component''),
    \item $A_R$ is a set of assignment edges%,  (in  the ``resource component''),
    \item $ASC$ is a set of (labeled) association edges,
    \item $P$ is a set of (labeled) prohibition edges,
    \item $V$ is a ``universe'' set of entities that can be in the model,
    \item $COM$ is a set of commands.
\end{itemize}
We require that the sets $U$, $UA$, $R$, $RA$ and $R_\psi$ are pairwise-disjoint and $U\cup UA\cup R\cup RA \subseteq V$\footnote{Alternative formulations are possible, e.g.\ $R\subset RA$ is assumed in~\cite{CDX23}, and the universe and commands are elided.}.
The set $U\cup UA\cup R\cup RA$ forms the vertex set of the digraph $G$ representing the state of the model. 
The access rights $R_\psi$ are used as labels for edges in $ASC$ and $P$, and the sets $A_U$, $A_R$, $ASC$ and $P$ form edges of $G$.
We define $G$ piece-by-piece as follows. 

The sets $A_U$ and $A_R$ correspond to unlabeled, directed edges in $G$ on the vertex sets $U\cup UA$ and $R\cup RA$ such that
\begin{align*}
    A_U &\subseteq (U \times UA)\cup (UA\times UA)\\
    A_R &\subseteq (RA\times R)\cup (RA\times RA).
\end{align*}
It is standard (e.g.~\cite{CDX21}) to assume that the subgraphs of $G$ given by $(U\cup UA, A_U)$ and $(R\cup RA, A_R)$ are acyclic and to refer to them as the \emph{user DAG} and the \emph{resource DAG} respectively. 
The set $ASC$ is a set of labeled edges satisfying
$ASC \subseteq UA \times RA \times R_\psi$, where we think of $a=( ua, rsa, r)$ as an edge in $G$ from the user attribute $ua$ to the resource attribute $rsa$ labeled with the access right $r$.
The prohibition edges $P$ satisfy $P \subseteq UA \times RA \times R_\psi$ and are also thought of as labeled edges of $G$ in the same way. 

The universe set $V$ represents the collection of all possible users, user attributes, resources, and resource attributes that may be added to the state of the model. 
That is, we can think of $V$ as a ``reservoir'' of vertices that do not yet exist in $G$ but that may be added by commands. We restrict $R_\psi$ and $V$ to be finite.

To reflect the dynamic nature of the NGAC model, wherein commands can change the state digraph $G$, it is convenient to index the state graphs with a time variable $t$. 
We call the initial state digraph $G_0$ and the execution of a command $c$ at time $t\ge 1$ takes $G_{t-1}$ as input and produces another state digraph $G_{t}$. 
The set $COM$ is a set of commands of the following form.
\begin{algorithmic}
    \Function{$\alpha$}{$X_1,X_2,\dotsc,X_k$}
        \If{$cond_1$ \textbf{and} $cond_2$ \textbf{and} $\dotsb$ \textbf{and} $cond_m$}
        \State $op_1$
        \State $op_2$
        \State $\dotsc$
        \State $op_n$
        \EndIf
    \EndFunction
\end{algorithmic}
where $X_1,\dotsc, X_k$ are the formal parameters of the command, $cond_i$ is a condition and $op_j$ is one of the primitive operations detailed in Table~\ref{op} (whose input must be one of the formal parameters of $\alpha$). 
We require that each condition is of the form $s \notin G_{t-1}$ where $G_{t-1}$ is the state digraph before the execution of the command and $s$ is either an entity in $V$ or a possible edge in a state digraph\footnote{i.e.\ a suitable element of $V\times V$ or $V\times V\times R_\psi$.}.
Each primitive operation consists of a single addition or deletion of an element of the sets $U$, $UA$, $R$, $RA$, $A_U$, $A_R$, $ASC$, and $P$.
Thus, there are 16 primitive operations.
Note that by convention, for each of the primitive operations that destroys a vertex $v$ of $G$ (i.e.\ an element of $U\cup UA\cup R\cup RA$) also destroys all edges of $G$ that involve $v$. 
For brevity, this convention is suppressed from the notation in Table~\ref{op}.

\begin{table*}[ht]
    \centering
    \caption{List of Primitive Operations}
    \label{op}
    \begin{tabular}{l l l}
        \toprule
        \textbf{Operation} & \textbf{Conditions} & \textbf{Action} \\
        \midrule
        \textbf{create user} $u$ & $u \not\in U \wedge u \in V$ & $U\mapsto U \cup \{u\}$\\
        
        \textbf{create user attr.} $ua$ & $ua \not\in UA \wedge ua \in V$ & $UA \mapsto UA \cup \{ua\}$\\
        
        \textbf{create res.} $rs$ & $rs \not\in R \wedge rs \in V$ & $R\mapsto R \cup \{rs\}$\\
        
        \textbf{create res.\ attr.} $rsa$ & $rsa \not\in RA \wedge rsa \in V$ & $RA \mapsto RA \cup \{rsa\}$\\
        
        \textbf{create user assign.} $au$ & $au \not\in A_U \wedge au \in (U \times UA)\cup (UA\times UA)$ & $A_U\mapsto A_U \cup \{au\}$\\
        
        \textbf{create res.\ assign.} $ar$ & $ar \not\in A_R \wedge ar \in (RA \times R)\cup (RA\times RA)$ & $A_R\mapsto A_R \cup \{ar\}$\\
        
        \textbf{create assoc.} $a$ & $a \not\in ASC \wedge a \in UA \times RA \times R_\psi$ & $ASC \mapsto ASC \cup \{a\}$\\
        
        \textbf{create prohib.} $p$ & $p \not\in P \wedge p \in UA \times RA \times R_\psi$ & $P\mapsto P \cup \{p\}$\\
        
        \textbf{destroy user} $u$ & $u \in U$ & $U \mapsto U \,\backslash\,\{u\}$\\
        
        \textbf{destroy user attr.} $ua$ & $ua \in UA$ & $UA \mapsto UA \,\backslash\,\{ua\}$\\
        
        \textbf{destroy res.} $rs$ & $rs \in R$ & $R\mapsto R \,\backslash\,\{rs\}$\\
        
        \textbf{destroy res.\ attr.} $rsa$ & $rsa \in RA$ & $RA\mapsto RA \,\backslash\,\{rsa\}$\\
        
        \textbf{destroy user assign.} $au$ & $au \in A_U$ & $A_U\mapsto A_U \,\backslash\, \{au\}$\\
        
        \textbf{destroy res.\ assign.} $ar$ & $ar\in A_R$ & $A_R \mapsto A_R \,\backslash\, \{ar\}$\\
        
        \textbf{destroy assoc.} $a$ & $a \in ASC$ & $ASC \mapsto ASC \,\backslash\, \{a\}$\\
        
        \textbf{destroy prohib.} $p$ & $p \in P$ & $P\mapsto P \,\backslash\, \{p\}$\\
        \bottomrule
    \end{tabular}
\end{table*}

Another standard assumption that we make is that all commands are \emph{mono-operational}, which means each command contains one primitive operation, i.e.\ $n = 1$ in the above description of a command. 
Since all primitive operations only take one input, then we can also say that we only need one formal parameter per command, i.e. $k = 1$. 
Thus, we can drop the subscripts on the formal parameter and operation in a command.
We can also assume that $m$ is finite since the fact that $V$ is finite implies that the number of possible conditions is also finite. 
Note that for each command, the set of possible values for the formal parameter is finite, since every set mentioned in the conditions is finite too.

In our algorithm (Section~\ref{sec:alg}), we require a few additional, mild assumptions. 
\begin{enumerate}
    \item\label{itm:onlyedgecreatehascond} Only commands $\alpha(X)$ whose single operation is of the form ``\textbf{create \ldots} $X$'', where $X$ is an edge of the state graph (i.e.\ a user assignment, resource assignment, association or prohibition), have any condition.
    \item\label{itm:onlyedgechecksincond} The conditions in such a command are of the form $e_1 = X$ to make sure the command can only be executed for the input $e_1$, or $e_2 \notin G_{t-1}$ where $e_2$ is a potential edge of the state digraph and $G_{t-1}$ is the state digraph before the execution of the command. 
    These conditions can prevent the creation of a specific edge due to the presence of existing edges.
    \item\label{itm:alldestroy} All vertices and edges which may be present in the state digraph have unconditional destroy commands.
\end{enumerate}

It is natural to consider conditions only on the creation of edges of the state graph as it is edges which determine access to resources and so guarding their creation with conditions is flexible enough to model real-world access policies. Since creating a node of the state graph with no adjacent edges cannot create undesirable access paths, it is not a significant restriction to demand that vertex creation is unconditional.

On the other hand, edge creation conditions allow us  to model natural constraints, for example that a user must have the attribute ``student'' in order to be allowed to gain the attribute ``teaching assistant''. 
This natural type of constraint is expressible within the restrictions we give in assumption~\ref{itm:onlyedgechecksincond}.

Assumption~\ref{itm:alldestroy} is also important in real-world access control models, as e.g.\ when a user leaves an organization one should delete all vertices in the state graph corresponding to that user. 
Though deleting vertices corresponding to attributes is not normally required, allowing for this does not significantly alter the NGAC model or its applicability.
The situation is similar with edges. Promotion or other role changes of personnel in an organization may require the deletion of edges in the state graph (e.g. the promotion of a user to a higher rank with greater security privileges). 
While association or prohibition edges are normally static, they may need to be deleted if the organization suddenly adopts stricter access policies. 

Despite these further restriction on the conditions in commands, NGAC models satisfying our assumptions are sufficiently expressive to allow for \emph{mutually exclusive attributes}. 
These are groups of attributes such that a user or resource can have at most one of the attributes in the group. 
Natural examples of mutually exclusive attributes abound. 
In the context of separation of duty (which is commonly discussed in the context of RBAC~\cite{LTZ07}) one might choose a ``role'' to be a group of mutually exclusive user attributes. An example due to Sandhu and Samarati~\cite{SS94} is that the ``authorizer'' and ``preparer'' roles in the context of paychecks should be mutually exclusive. 
The system is more vulnerable to abuse if a single user has the right to both prepare and authorize checks. If we separate those duties then coordination is required for abuse to occur.
More generally, in ABAC there are natural attributes that users cannot hold simultaneously. 
For example in the category of ``age'' the user attributes ``minor'' and ``adult'' are mutually exclusive. 
Resource attributes can also be mutually exclusive, for example a ``security clearance'' could be exactly one of ``top secret'', ``secret'', ``confidential'' or ``public''. 

Since assigning users or resources roles takes the form of creating edges from $U$ or $R$ to $UA$ or $RA$ in the state digraph, we can implement mutually exclusive attributes in the presence of the assumptions~\ref{itm:onlyedgecreatehascond}--\ref{itm:alldestroy} with create commands whose conditions check that the user or resource which is to gain an attribute does not already have an attribute in the mutually exclusive group. 
We will see in Section~\ref{sec:real} that natural groups of mutually exclusive attributes are a significant obstacle for the algorithm we propose in Section~\ref{sec:alg}.

\subsection{The safety problem}

The following definition simplifies the study of safety in the NGAC model. 

\begin{definition}%[Access relation]
    Given a state digraph, we define for each $r\in R_\psi$ the \emph{access relation} to be binary relation $\rightarrow^r$ (formally, a subset of $U \times R$) such that $u \rightarrow^r rs$ if and only if there exists a path in the state digraph from $u$ to $rs$ that goes through an association edge labeled with $r$.

    When studying an NGAC model the state graph is dynamic and hence for each $r\in R_\psi$, the relation $\rightarrow^r$ is also dynamic. 
    Given a sequence $S$ of commands, we write $u\rightarrow_{S}^rrs$ to mean that starting from the current state digraph, after applying the sequence of commands $S$ we have $u\rightarrow^rrs$ in the resulting state digraph.
\end{definition}

\noindent
Note that there are no primitive operations that modify $R_\psi$ and so we cannot create or destroy access relations over time, only modify them. 
We can now define the concept of safety that we study. 

\begin{definition}\label{def:safety}
    Given an NGAC model $M$, we say that $M$ is \emph{safe} if, for all rights $r \in R_\psi$, and all finite sequences of commands $S$ whose elements are in $COM$, we have that 
    \[\rightarrow_{S}^r {} \subseteq {}\rightarrow^r. \]
\end{definition}

\noindent 
The subset notation above means that after the execution of $S$, there cannot be a new element of any access relation. 
That is, no matter what sequence of commands is performed on the model goes through, no new access is gained.

\begin{definition}
    The \emph{safety problem} $\SP$ is the decision problem that, given an input NGAC model $M$, returns ``Yes'' if the model is safe in the sense of Definition~\ref{def:safety}.

    The \emph{co-safety problem} $\coSP$ is the complement of $\SP$. That is, $\coSP$ is a decision problem and the answer to $\coSP(M)$ is ``Yes'' if and only if the answer to $\SP(M)$ is ``No''.

    Note that throughout, we assume that the NGAC model input to these safety problems satisfies the assumptions outlined above. In particular, $V$ is finite, and the commands in $COM$ are mono-operational with the structure required above.
\end{definition}

\section{Computational Complexity}\label{sec:complexity}

In this section we prove that the co-safety problem is $\NP$-complete. 
We will first define a graph-theoretic abstraction of co-safety that we call \emph{directed acyclic constrained connectivity} ($\DACC$). 
This problem serves as an intermediate step is our reduction that ``smooths out'' several features of the NGAC model that are not critical to the complexity of the safety problem. Recall that we show 
\[\COL \leq_p \DACC \leq_p \coSP, \]
where $\COL$ is the classic 3-coloring decision problem (see e.g.~\cite[GT4]{GJ79}).

\begin{definition}%[Constraint graph]
    Given a directed acyclic graph $\Gamma = (V, E)$, a \emph{constraint graph} for $\Gamma$ is a graph $C = (E(\Gamma), E')$ where $e_1e_2 \in E'$ represents the constraint that edges $e_1$ and $e_2$ cannot both exist.

    Given a constraint graph $C$ for $\Gamma$, we say that a subgraph $\Gamma'\subseteq \Gamma$ is \emph{valid} if the edges of $\Gamma'$ form an independent set in $C$. That is, at most one edge of $\Gamma$ from each constraint can be present in $\Gamma'$
\end{definition}

\noindent
The $\DACC$ problem is about the existence of paths in valid subgraphs of a directed acyclic graph $G$.

\begin{definition}
    The \emph{directed acyclic constrained connectivity} problem ($\DACC$) is the decision problem which, given as input a directed acyclic graph $\Gamma = (V, E)$, a constraint graph $C$, and a pair of vertices $s, t \in V$, asks whether there exists a valid subgraph of $\Gamma$ in which there is a (directed) path from $s$ to $t$.
\end{definition}

\noindent One can view $\DACC$ as a modification of the classic $st$-connectivity problem where we introduce constraints on the path from $s$ to $t$ that one should find: we restrict our attention to paths that are valid subgraphs of $\Gamma$ given some constraint graph $C$. 
Our first result is that $\DACC$ is $\NP$-complete.

\begin{theorem}\label{thm:DACCNPcomp}
    $\DACC$ is $\NP$-complete.
\end{theorem}
\begin{proof}
    It is easy to see that $\DACC$ is in $\NP$ as one can take the certificate of a ``Yes'' answer to be the $st$-path itself, and one can easily check validity in polynomial time by looping through the constraints. 

    It remains to show that $\DACC$ is $\NP$-hard, which we do by giving a reduction $\COL\lep\DACC$.

    Let $G$ be an input to $\COL$, and let $n = |V(G)|$. We construct an input $(\Gamma, C, s, t)$ for $\DACC$ such that $G$ is 3-colorable if and only if $(\Gamma, C, s, t)$ yields the answer ``Yes'' in $\DACC$. 
    A key property will be that $st$-paths in $\Gamma$ correspond to 3-colorings (not necessarily proper) of the vertices of $G$. 

    First, let $V(G) = \{v_1,\dotsc,v_n\}$ under some arbitrary ordering. 
    To construct $\Gamma$ we include a copy of each $v_i$ in $V(G)$ and, for each $i$, three vertices labeled $R_i, G_i,$ and $B_i$. Then, we create (directed) edges from each $v_i$ to each of $\{R_i, G_i, B_i\}$, and for $i<n$ from each $\{R_i, G_i, B_i\}$ to $v_{i+1}$. Finally, we create new vertices $s$ and $t$, connect $s$ to $v_1$ and each of $\{R_n, G_n, B_n\}$ to $t$. This process takes time $\mathcal{O}(n)$.

    The DAG $\Gamma$ is built in such a way that an edge from $v_i$ to $R_i$, $G_i$ or $B_i$ can be used to represent an assignment of $v_i$ to the respective colors red, green, and blue. We add constraints to $C$ such that for each $i$
\begin{enumerate}
    \item\label{itm:singlassign} at most one of the edges $v_iR_i$, $v_iG_i$, and $v_iB_i$ can exist. 
    \item\label{itm:nonmono} For adjacent vertices $v_iv_j$, at most one of the edges $v_iR_i$ and $v_jR_j$ exist, and similar for the colors $G$ and $B$.
\end{enumerate}

The vertex set of $C$ is simply the edge set of $\Gamma$. 
In time $\mathcal{O}(n)$ we can add the ``single assignment constraints'' described in~\ref{itm:singlassign}. 
Note that this comprises the formation of triangles on the triples of vertices $\{v_iR_i,v_iG_i,v_iB_i\}$ for each $i$.
Similarly, in time $\mathcal{O}(n^2)$ we can add the ``coloring constraints'' described in~\ref{itm:nonmono}. 

Now that we have our input $(\Gamma, C, s, t)$ constructed in polynomial time, we must show that ``Yes'' for $G$ in $\COL$ corresponds to ``Yes'' for $(\Gamma, C, s, t)$ for $\DACC$.

Let $G$ be 3-colorable, and let $\varphi: V(G) \rightarrow \{R, G, B\}$ a proper coloring of $G$. We show that the path 
\[ P = (s, v_1, \varphi(v_1), v_2, \varphi(v_2),\dotsc, v_n, \varphi(v_n), t) \]
is valid in $\Gamma$. 
By construction $P$ cannot violate any of the single assignment constraints, and since $\varphi$ is proper it does not violate any of the coloring constraints.

Suppose that $G$ has no proper $3$-coloring. 
Then every path $P$ from $s$ to $t$ in $\Gamma$ is invalid, because each such path corresponds to a function $\varphi : V(G)\to\{R,G,B\}$ in the natural way. 
Since $G$ is not 3-colorable there must exist a pair $v_iv_j$ of adjacent vertices which get the same color under $\varphi$, but this means that the edges $v_i\varphi(v_i)$ and $v_j\varphi(v_j)$ are both in the path $P$. But this violates one of the coloring constraints. 
\end{proof}

Turning to the safety problem, we show $\coNP$-completeness in two steps. 

\begin{lemma}\label{lem:coSP-NP}
    $\coSP\in\NP$.
\end{lemma}
\begin{proof}
    If an NGAC model $M$ (with the standard 11-tuple notation) is a ``Yes'' instance of $\coSP$, then there is a tuple $(u, rs, r) \in U \times R \times R_\psi$ and a sequence of commands $S$ such that $u \nrightarrow^r rs$ and $u \in\rightarrow_{S}^r rs$.
    We let the certificate be a combination of the tuple and the sequence $S$.
    To verify the certificate, we need to check both the conditions. 

    To check that $u \nrightarrow^r rs$ we start with the state digraph $G_0$ and remove all edges in $ASC$ that are not labeled $r$, then use breadth-first search to check that there is no path from $u$ to $rs$. This can be done in time polynomial in the size of the input $M$ to the $\coSP$ problem.

    To check that $u \in\rightarrow_{S}^r rs$, we start with $G_0$ and apply the sequence $S$ of commands. Since each command and possible input it can take are part of the input $M$, this takes time polynomial in the size of $M$. Suppose that this gives the state digraph $G_t$.
    Then, we can remove all edges in the the $ASC$ part of $G_t$ that are not labeled $r$, and use breadth-first search to check that there is a path from $u$ to $rs$ in this graph. Again, this takes polynomial time.
\end{proof}

We now give a reduction showing hardness.
The key idea is that we can build an NGAC model $M$ in which commands add or remove edges of the state digraph in such a way that the constraints in an instance $(\Gamma,C,s,t)$ of $\DACC$ are respected. 
Figure~\ref{fig:cc to safety reduction} shows the basic idea, where we embed valid subgraphs of $\Gamma$ in the subgraph of the state digraph induced by the user attributes $UA$. 
That is, we embed a hard instance $(\Gamma,C,s,t)$ of DACC in the subgraph of the state digraph of the NGAC model induced by the user attributes $UA$. 
It suffices construct an NGAC model with a single user $u$ connected by an edge to the user attribute represented by the source $s$ in $\Gamma$, and a single edge from the target $t$ of $\Gamma$ to a single resource attribute $rsa$ (labeled by a single right $r$). Finally, we have a single resource $rs$ to which $rsa$ is connected by an edge. 
The rest of the proof consists of showing that with a suitably-defined set of commands $COM$, the states of the model we can explore by running commands correspond to valid subgraphs of $\Gamma$. 

\begin{theorem}
    $\coSP$ is $\NP$-complete.
\end{theorem}

\begin{proof}
    Sine we have Lemma~\ref{lem:coSP-NP} and Theorem~\ref{thm:DACCNPcomp}, it suffices to show that $\DACC\lep\coSP$. 
    
    Let $(\Gamma, C, s, t)$ be the input of $\DACC$. 
    We construct an NGAC model $M$ with the standard 11-tuple notation as follows. 
    Let $U = \{u\}$, $UA = V(\Gamma)$, $R = \{rs\}$, $RA = \{rsa\}$, $A_U = \{(u, s)\}$, $A_R = \{(rsa, rs)\}$, $ASC = \{(t, rsa, r)\}$, $P = \emptyset$, $R_\psi = \{r\}$, $V=\{u,rs,rsa\}\cup V(\Gamma)$, and let $COM$ be the set of commands defined as follows. 
    
    First, in $COM$ we include an unconditional destroy command for every vertex and and edge of $\Gamma$. Then, for every constraint $e_1e_2 \in E(C)$, we add to $C$ the two commands

    \begin{algorithmic}
    \Function{$\alpha_{1,2}$}{$X$}
        \If{$e_1 == X$ \textbf{and} $e_2\notin G_{t-1}$}
        \State \textbf{create user assign.} $X$
        \EndIf
    \EndFunction

    \Function{$\alpha_{2,1}$}{$X$}
        \If{$e_2 == X$ \textbf{and} $e_1\notin G_{t-1}$}
        \State \textbf{create user assign.} $X$
        \EndIf
    \EndFunction
    \end{algorithmic}

    The conditions in the command ensure that when starting from a state in which at most one of $e_1$ and $e_2$ exist at the same time, then this remains true. 

    This gives an input $M$ to $\coSP$, one that can clearly be constructed in polynomial time in the size of $\Gamma$.

    Given the command set $COM$, it is straightforward to see that the valid subgraphs of $\Gamma$ are in bijection with possible states of the user DAG, the subgraph of the state digraph induced by the user attributes $UA$.
    Moreover, we have $u\nrightarrow^r rs$ in the initial state graph $G_0$ by construction.
    It remains to show that the answer for $(\Gamma,C,s,t)$ in $\DACC$ is ``Yes'' if and only if the answer for $M$ in $\coSP$ is ``Yes''. 

    Let $(\Gamma, C, s, t)$ be a ``Yes'' instance of  $\DACC$. Then there is a valid subgraph $\Gamma'$ of $\Gamma$ in which there is a path from $s$ to $t$. 
    But this means that there is a sequence of commands $S$ that results in the state digraph $G_t$ having a path from $u$ to $rs$ through an edge in $ASC$ labeled $R$. 
    That is, we have $u \rightarrow_S^r rs$ and hence $M$ gives a ``Yes'' in $\coSP$.
    
    Let $(\Gamma, C, s, t)$ be a ``No'' instance of $\DACC$.
    Then, there is no valid subgraph of $\Gamma$ that connects $s$ and $t$. But the conditions on edge creation commands in $COM$ mean that the only valid states of $M$ are digraphs in which the subgraph induced by $UA$ is a valid subgraph of $\Gamma$. 
    Thus, there is never an $st$-path in a state digraph of $M$ and hence there is never an access relation of the form $u\rightarrow_S^r rs$ for any sequence $S$ of commands. Then $M$ is a ``No'' instance of $\coSP$ as required. 
\end{proof}

\begin{figure}[!htb]
    \centering
    \begin{tikzpicture}[scale=0.75]

    \draw[thick] (0,0) circle (2cm);

    \node at (0, 0) {$\Gamma$};

    \filldraw[black] (-1,1) circle (2pt);
    
    \node[anchor=south] at (-1,1.1) {$s$};

    \filldraw[black] (-3,3) circle (2pt);
    
    \node[anchor=south] at (-3,3.1) {$u$};

    \filldraw[black] (1,-1) circle (2pt);
    
    \node[anchor=south] at (1,-0.9) {$t$};

    \filldraw[black] (5,-1) circle (2pt);
    
    \node[anchor=west] at (5,-0.9) {$rsa$};

    \filldraw[black] (5,3) circle (2pt);
    
    \node[anchor=south] at (5,3) {$rs$};

    \draw[->] (-3,3) -- (-1.1,1.1);
    \draw[dashed, ->] (1,-1) -- (4.9,-1) node[midway, above] {$r$};
    \draw[->] (5,-1) -- (5,2.9);

\end{tikzpicture}
    \caption{Reduction from $\DACC$ to $\coSP$.}
    \Description[A picture of the reduction from $\DACC$ to $\coSP$]{A picture of the reduction from $\DACC$ to $\coSP$ in which valid subgraphs of $G$ are embedded in the subgraph of the state digraph of the NGAC model induced by the user attributes.}
    \label{fig:cc to safety reduction}
\end{figure}
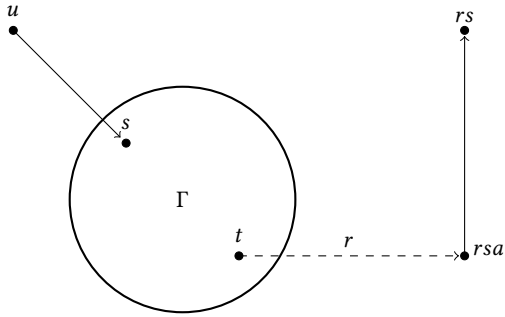

\section{Algorithm}\label{sec:alg}

In this section, we will introduce an algorithm that solves the safety problem for the NGAC model significantly faster than naive brute-force search. 
We will first give an algorithm to solve $\DACC$, and then explain how we can transform a safety problem input to an input for $\DACC$ to use the algorithm.

\subsection{Solving \texorpdfstring{$\DACC$}{DACC}}

Let $(G, C, s, t)$ be an input to $\DACC$. Since we are trying to find a path in a valid subgraph of $G$, it suffices to enumerate \emph{maximal} valid subgraphs of $G$. 
This is because adding edges to a valid subgraph cannot destroy the path we are looking for. 
In terms of the constraint graph $C$, this means it suffices to enumerate the maximal independent sets $\mu(C)$ of $C$. 
An \emph{output sensitive} enumeration algorithm for a set $S$ with delay $d$ takes time at most $d|S|$ to enumerate $S$. We think of the delay as an upper bound on the computation time needed per element for the numeration. 
Our algorithm simply uses a state-of-the-art output-sensitive enumeration algorithm for MIS~\cite{MU04}\footnote{The algorithm \textsc{AllMaxCliques} in particular, though MIS in $G$ are maximal cliques in the complement of $G$ so we have to take the complement first.} and checks for the existence of an $st$-path in each of the returned MIS of $C$ (which corresponds to a valid subgraph of $G$) with breadth-first search. 
See Algorithm~\ref{alg:DACC}.
The algorithm we use~\cite{MU04} for enumerating MIS has a time delay of $M(n)$ when the input is an $n$-vertex graph, where $M(n)$ is the time needed to multiply two $n\times n$ matrices. The best available result~\cite{ADW+25} shows that $M(n) = \mathcal{O}(n^{2.372})$ (though exponents approximately $2.37$ have been known for some time).

\begin{algorithm}[H]
    \caption{An algorithm for $\DACC$ based on enumeration of MIS.}
    \label{alg:DACC}
    \begin{algorithmic}[1]
        \Procedure{DACC-via-MIS}{$G = (V, E), C, s, t$}
        \For{$I \in \mu(C)$}\label{line:mu} \Comment{Use the algorithm of~\cite{MU04}}
            \State $G' \gets (V, I)$
            \If{$\exists$ $st$-path in $G'$}\label{line:BFS} \Comment{Use breadth-first search}
                \State \Return True
            \EndIf
        \EndFor
        \State \Return False
        \EndProcedure
    \end{algorithmic}
\end{algorithm}

A well-known result in extremal combinatorics~\cite{MM60,MM65} shows that the number of MIS in an $n$-vertex graph is at most $3^{n/3} < 1.45^n$. The result is often attributed to Moon and Moser~\cite{MM65}, but was proved independently by Miller and Muller~\cite{MM60} several years before. 
Simpler proofs were given more recently by Vatter~\cite{Vat11} and Wood~\cite{Woo11}.
This result forms the basis of our analysis of the worst-case time complexity of our algorithm for $\DACC$. 
The upper bound is achieved by a disjoint union of triangles, which we refer to later in our consideration of whether real-world examples are close to the worst case of this algorithm.
Overall, in the worst case this enumeration of MIS significantly outperforms a naive search over all sets (filtering out the non-independent ones), and more importantly there are classes of graphs with a polynomially-bounded number of maximal independent sets and thus the algorithm takes polynomial time on such graphs. 
We bound the running time of Algorithm~\ref{alg:DACC} below.

\begin{lemma}\label{lem:DACCruntime}
    Let $(G,C,s,t)$ be an input to Algorithm~\ref{alg:DACC} such that $n=|V(G)|$, $m=|E(G)|$ and $\mu(C)$ is the set of MIS in $C$. 
    Then the running time is 
    \[ \mathcal{O}(n+m)\cdot M(m) \cdot \mu(C) \le \mathcal{O}(n^{6.75})\cdot \mu(C)\le \mathcal{O}(1.45^{n^2}).\]
\end{lemma}

\begin{proof}
We have an input $(G,C,s,t)$ for the $\DACC$ problem with $n=|V(G)|$ and $m=|E(G)|$. 
Algorithm~\ref{alg:DACC} is simple and loops over the MIS of $C$, executing for each MIS a breadth-first search in a sungraph of $G$. 
This search takes time $\mathcal{O}(n+m)$ (the worst-case time complexity of standard breadth-first search).
Since $|V(C)|=|E(G)|=m$, the algorithm~\cite{MU04} takes time $M(m)\cdot \mu(C)$ to do the required enumeration. 
This is because $M(m)$ is the time delay and $\mu(C)$ is the size of the set being enumerated.
Note that $m\le n^2$ and $M(m)\le \mathcal{O}(m^{2.372})$ to see the first inequality in the running time bound. For the final inequality note that $\mu(C)\le 3^{m/3} < 3^{n^2/3} < 1.45^{n^2}$.
\end{proof}

\subsection{Solving the safety problem}\label{section: solving sp}\label{sec:solvesaftey}

Let $M = (U, UA, R, RA, R_{\psi}, A_U, A_R, ASC, P, V, COM)$ be an input to the safety problem $\SP$ and that $M$ satisfies our assumptions. 
Recall that in addition to having mono-operational commands, we require the enumerated assumptions~\ref{itm:onlyedgecreatehascond}--\ref{itm:alldestroy} stated in Section~\ref{sec:prelim}.

Let $N$ be the length of the input $M$ in a natural binary encoding.
We will construct an input $(\Gamma, C, s, t)$ for $\DACC$ in time polynomial in $N$, and use algorithm \ref{alg:DACC} to solve $\DACC$ for $(\Gamma, C, s, t)$ in order to solve the safety problem for $M$. 

\subsubsection{Constructing \texorpdfstring{$\Gamma$}{Gamma}}

We let $\Gamma$ be the \emph{supergraph} of $M$, meaning that $\Gamma$ is the result of starting with the initial state digraph of $M$ and executing, for all valid inputs, all possible commands in $COM$ whose primitive operation is one of the form ``\textbf{create} \ldots'', \emph{without} checking the condition(s) in the command.
Since we assume that $COM$ is finite, and $V$ and hence the vertex set of $\Gamma$ is also finite, and we know that the number of edges of $\Gamma$ is bounded by polynomial in $|V|$ and $|R_{\psi}|$, it takes time polynomial in the size of the input $M$ (which includes $COM$, $R_{\psi}$ and $V$) to construct $\Gamma$. 
We also have the property that $G_t \subseteq \Gamma$ where $G_t$ is the state digraph at any given time $t$. 

\subsubsection{Constructing \texorpdfstring{$C$}{C}}\label{sec: construct C}

We construct a constraint graph $C$ for the supergraph $\Gamma$ which enforces the conditions of the create commands in the model $M$. 
The goal is to do this in such a way that valid subgraphs of $\Gamma$ correspond to possible state digraphs $G_t$ of the model $M$.

We initialize an empty constraint graph $C = (E(\Gamma), \emptyset)$. 
By assumptions~\ref{itm:onlyedgecreatehascond} and~\ref{itm:onlyedgechecksincond}, each edge creation command has one parameter $X$ corresponding to an edge of $\Gamma$ and must have a conditional of the form 
\[ e == X \mathop{\textbf{and}} cond_1 \mathop{\textbf{and}} cond_2 \mathop{\textbf{and}} \dotsb \mathop{\textbf{and}} cond_m, \]
where each $cond_i$ is of the form $e_i\notin G_{t-1}$ for some potential edge (element of $V\times V$ or $V\times V \times R_\psi$).
We may assume that $e_i\in\Gamma$, else the command can never execute and we may remove it.
Then for each such command $c$ with this structure we add the edges 
\[ \{ ee_i : i\in M\} \]
to $C$. 
This will make sure that in a valid subgraph $\Gamma'\subseteq \Gamma$, if $e$ exists in $\Gamma'$ then the edges in the conditions of $c$ must not exist. 

\begin{lemma}\label{lem:validsubgraphs}
    The valid subgraphs $\Gamma'$ of $\Gamma$ correspond to possible states of the model $M$.
\end{lemma}
\begin{proof}    
    Let $G_0$ be the initial state digraph of $M$. Writing 
    \[ M= (U, UA, R, RA, R_{\psi}, A_U, A_R, ASC, P, V, COM), \] we have
    \[ G_0 = (U \cup UA \cup R \cup RA, A_U \cup A_R \cup ASC \cup P).\] 

    We first prove that any valid subgraph $\Gamma'$ of $\Gamma$ is a valid state of the model. Given $\Gamma'$, we want to sequence of commands $S$ which, starting from $G_0$, puts the model in state $\Gamma'$.

    First, for all vertices and edges that are in $G_0$ but not in $\Gamma'$, we have to add a command to $S$ to destroy these elements of the state digraph. Note that by assumption~\ref{itm:alldestroy}, all vertices and edges must have an unconditional destroy command, thus we are able to destroy the necessary vertices and edges.
    
    Next, we have to create all the vertices and edges that are in $\Gamma'$ but not in $G_0$. Since $\Gamma$ was constructed to be the supergraph of $M$, for every necessary vertex or edge creation there is a command to create the required structure in the model. 
    If this command is unconditional then we simply add it to the sequence $S$. 
    By assumption~\ref{itm:onlyedgecreatehascond}, only edge creation commands have conditions, and by assumption~\ref{itm:onlyedgechecksincond} these conditions merely check for the absence of some edges. But by the construction of $C$ and the fact that $\Gamma'$ is valid, the conditions of any create command that we need to add to $S$ will be true at the time of execution. 
    
    Now we prove that any state of the model is a valid subgraph of $\Gamma$. Let $S$ be a sequence of commands and suppose that running $S$ from $G_0$ yields the state digraph $G_t$.
    
    By the construction of the command set $COM$, it is easy to prove (formally, by induction on the length of the sequence of commands $S$) that $G_t$ satisfies the constraints encoded by $C$. 
    The key to the proof is the construction of the constraints in $C$ from the commands themselves.
\end{proof}

\subsubsection{Testing for access paths}

Now that we have constructed $\Gamma$ and $C$, we want to loop over all possible access paths and check that, for each path that does not exist in the $G_0$, there is no sequence of commands $S$ that results in the access path existing. 

First, we fix a tuple $(u, rs, r) \in U \times R \times R_\psi$.
Then, if $u \not\rightarrow^r rs$ we run let $\Gamma^r$ be the subgraph of $\Gamma$ obtained from $\Gamma$ by deleting all association edges not labeled with $r$ and run Algorithm~\ref{alg:DACC} with input $(\Gamma^r, C, u, rs)$. 
If the answer is ``Yes'', then there must exist a sequence of commands $S$ such that $u \rightarrow_{S}^r rs$, which means we can return ``No'' for the safety problem with input $M$.

If all queries $\DACC(\Gamma^r, C, u, rs)$ in the loop above tuple return ``No'', then $M$ is safe and we return ``Yes'' in the safety problem.

\subsection{Analysis of running time}

Constructing $\Gamma$ takes time at most 
\[ \mathcal{O}(|COM|\cdot |V|^2\cdot |R_\psi|) \]
as the bulk of the construction occurs in a loop over each possible labeled edge of state the digraph in the model, where we check for a create command for that edge and add it to $\Gamma$.

Constructing $C$ takes time at most 
\[ \mathcal{O}(|COM|\cdot m_{\max} \cdot |O|^2\cdot |R_\psi|), \]
where $m_{\max}\le \mathcal{O}(|V|^2|R_\psi|)$ is the maximum number of conditions in any command in $COM$

Testing for access paths takes time at most 
\[ |O|^2|R_\psi| \cdot \mathcal{O}(|V|^2)\cdot M(|V|^2) \cdot \mu(C)  \]
because we execute for each tuple $(u, rs, r) \in V\times V\times R_\psi$ Algorithm~\ref{alg:DACC} with the input $(\Gamma^r, C, u, rs)$.
This $\Gamma^r$ has at most $|V|$ vertices and $|V|^2$ edges, and hence $\mu(C)\le 1.45^{|V|^2}$

Therefore the total runtime is at most 
\[ \mathrm{poly}(|M|)\cdot1.45^{|V|^2},\]
where $|M|$ is the size of the input $M$ and $|V|$ is the size of the set $V$ of possible objects.

\section{Real-world examples and MIS}\label{sec:real}

In this section we observe that real-world NGAC models can elicit nearly worst-case running time of the algorithm to solve the safety problem described in Section~\ref{sec:alg}. 
The main contributor to the running time is the number $\mu(C)$ of maximal independent sets in the constraint graph $C$ that is built from the commands of the model which enforce separation of duty. 

In practice, we often have small groups of mutually exclusive attributes. 
For example, the attributes \emph{teacher}, \emph{student}, and \emph{staff} might be (pairwise) mutually exclusive in an NGAC model designed for access control in an educational setting. 
To enforce this fact in the model we would have edge creation commands that will only create e.g.\ the edge from a user $u$ to the user attribute \emph{teacher} if neither the edges from $u$ to \emph{student} nor from $u$ to \emph{staff} are present. 
This leads to a triangle in $C$ on the vertices corresponding to the three mutually exclusive attributes. 
More generally, a set of $k$ mutually exclusive attributes leads to a clique of size $k$ in $C$.
If there are no other attributes that are incompatible with the set of mutually exclusive ones then this clique will be separated from the rest of the graph. 
So, for each set $K$ of mutually exclusive and otherwise non-interacting attributes, we see a clique in $C$ disjoint and unconnected to the rest of the graph.

But a disjoint union of small cliques leads to nearly the worst case for the number of MIS in a graph. Indeed, the results discussed in the introduction~\cite{MM60,MM65} show that a disjoint union of \emph{triangles} is actually the worst case. As seen from the above example, a common organizational structure leads to the creation of a small clique in $C$, so we can expect that typical NGAC models present somewhat challenging inputs for our algorithm to handle. 

Graphs of very high minimum degree have a polynomial number of MIS, but we do not expect typical organizational structures to yield this kind of constraint graph. This would require that all attributes are only compatible with a constant-sized set of other attributes. In the presence of user and resource attributes this seems very unlikely.

\section{Conclusions and future work}
In this paper, we analyzed the computational complexity of the Safety Problem in NGAC models. We show that under mild assumptions, the Safety Problem in NGAC models is $\coNP$-Complete. We also propose an algorithm to solve the Safety Problem that exploits the combinatorial nature of this problem which outperforms naive brute-force algorithms. We show that mutually exclusive attributes drive the computational complexity of the problem and the running time of our algorithm.

In future work, it would be interesting to test our algorithm with real-world or synthetic data to see if the safety problem tends to be intractable on realistic instances.
It would also be interesting to study the fine-grained aspects of complexity subject to parameters that control the number of groups of mutually exclusive attributes and their size.

\begin{acks}
The authors would like to thank the anonymous referees for
their comments and suggestions. 

This work is supported by the \grantsponsor{US-NSF}{U.S.\ National Science Foundation}{https://www.nsf.gov} under Grant
No.~\grantnum{US-NSF}{CCF-2309707}
and the \grantsponsor{US-NIST}{U.S.\ National Institute of Standards and Technology}{https://www.nist.gov} under Grant No.~\grantnum{US-NIST}{60NANB23D152}.
\end{acks}

\printbibliography

\end{document}